\documentclass[a4paper,11pt,onecolumn,tightenlines,superscriptaddress]{article}
\pdfoutput=1
\usepackage{preamble}
\usepackage[a4paper,left=3cm,right=3cm]{geometry}

\begin{document}
\title{Lattice moments and quantum designs}
\author{Jonathan Conrad\textsuperscript{$*$}, Kristian Holm\textsuperscript{$\dagger$}}
\date{\today}

\maketitle
{\renewcommand{\thefootnote}{}\footnotetext{%
  \hspace{-1.8em}\textsuperscript{$*$}{\'E}cole Polytechnique F{\'e}d{\'e}rale de Lausanne, email: \href{jonathan.conrad@epfl.ch}{jonathan.conrad@epfl.ch}
  \\
  \textsuperscript{$\dagger$}Christian-Albrechts-Universit{\"a}t zu Kiel, email: \href{holm@math.uni-kiel.de}{holm@math.uni-kiel.de}}}
\begin{abstract}
    We show that the ensemble of unimodular Gottesman--Kitaev--Preskill distributions fails to form a rigged quantum $3$-design and quantify the replica defect that determines this failure. This is established by rederiving a second moment formula for symplectic lattices in a self-contained manner.
\end{abstract}

\section*{Introduction}
Gottesman--Kitaev--Preskill (GKP) codes fix tempered distributions associated to infinite-dimensional Hilbert spaces $\CH=L^2(\R^n)$ with a symplectic lattice symmetry by $\Lambda\subseteq\Lambda^{\perp}\subset \R^{2n}$. In~\cite{Conrad_2025} it was shown that uniform ensembles of such distributions over the space of symplectic lattices $Y_n$ and their fundamental translations $z\in\CP(\Lambda)$ yield infinite-dimensional quantum $2$-designs. Here we show how they fail to also yield $3$-designs by proving the following.

\begin{theorem}\label{thm:main}
For $n\geq2$, the ensemble of unimodular GKP states, chosen uniformly over $\Lambda\in Y_n$ and $z\in \CP(\Lambda)$, fails to form a rigged quantum $3$-design relative to Schwartz operators. Its third moment, relative to Schwartz operators, is given by
\begin{equation}
    M_3=6\Pi_3+\CD_3,\quad 
    \CD_3=\frac{1}{2\zeta(2n)}
    \sum_{\substack{a, b\in \Z^{\times}\\ a+b\neq0}}Q_{ab}^{-n}U_{S_{ab}}(I+P_{(12)}).
\end{equation}
 $U_{S_{ab}}$ is a unitary representation for matrices $S_{ab}\in \Sp_6(\R)\cap \SO_6(\R)$, $P_{(12)}$ the permutation operator of tensor factors $(1, 2)$, and $Q_{ab}=a^2+b^2+ab$.
For the $n$-mode vacuum state $\ket{0}$, it holds that
\begin{equation}
     \|\CD_3\|=   \langle0|^{\otimes3}\CD_3|0\rangle^{\otimes3}
    =\frac{\zeta_{{\sf A}_2}(n)}{\zeta(2n)}-6= 6\cdot  3^{-n}+O(7^{-n}),
\end{equation}
where $\zeta_{{\sf A}_2}(n)=\sum_{(a,b)\neq(0,0)}(a^2+b^2+ab)^{-n}$ is Epstein's $\zeta$-function for ${\sf A}_2$.
\end{theorem}
For $r_{a,b}=(a,b,-a-b)^\intercal \in {\sf A}_2$, the replica reflection 
and its phase-space action are given by
\begin{equation}
 R_{a,b}=I_3-\frac{2r_{a,b}r_{a,b}^\intercal }{r_{a,b}^\intercal r_{a,b}},\qquad
 S_{ab}=R_{a,b}\oplus R_{a,b} \in\Sp_6(\R)\cap\SO_6(\R),
\end{equation}
where the symplectic map $S_{ab}$ acts as a Gaussian unitary transformation $U_{S_{ab}}\psi(q)=\psi((R_{a,b}\otimes I_n)q)$.
 GKP codes \cite{Gottesman_2001} are quantum stabilizer codes for an infinite-dimensional separable Hilbert space $\CH=L^2(\R^{n})$, where the stabilizer group is chosen as a subgroup of the unitary representation of the Heisenberg group $H(\R^n)$ via displacement operators
\begin{align}
    D:\R^{2n}&\to U(\CH),\quad D(x)D(y)=e^{-i\pi\omega(x,y)}D(x+y),\label{eq:disp_intro}
\end{align}
for a symplectic form $\omega:\, \R^{2n}\times\R^{2n}\to \R$. A GKP code is presented via a lattice basis $M=(z_1,\hdots,z_{2n})^\intercal $ for a symplectically integral lattice, which labels commuting displacements and a choice of stabilizer phases that together fix a finite-dimensional subspace $\CC_M\subset\CS'(\R^n)$. 

While the fixed subspace (the code space) $\CC_M\subset\CS'(\R^n)$ is unphysical, approximations of $\CC_M$ in $L^2(\R^n)$ have been shown to inherit a sufficient amount of structure of $\CC_M$ to be useful in the realization of quantum error correction in theory as well as experiments, e.g.~\cite{Sivak_2023}. They have further found application in quantum sensing and tomography; we refer the curious reader to the large body of literature on GKP codes for more, see e.g.~\cite{Brady_2024} and~\cite{conradPhD}. In this work, we focus on their application as quantum designs, studied in~\cite{Conrad_2025}, which shows that for the set of all unimodular GKP distributions $G_{\Lambda, s;M}$, for which $\Lambda=\Lambda^{\perp}$ is a symplectic lattice\footnote{Here and in the remainder of the paper, we understand that the dual $\Lambda^\perp$ is taken with respect to the symplectic bilinear form $\omega$ on the ambient vector space, so that a symplectic lattice satisfies $\Lambda \subseteq \Lambda^\perp$ by definition.}  and $s\in \C^n/\Lambda$ is a modular shift, satisfies that, for $t=2$,
\begin{equation}
    M_t = \int_{\chi_n} G_{X}^{\otimes t}  dX = a_t\Pi_t.
\end{equation}
Here the integration is carried out uniformly over the set $\chi$ of symplectic lattices $Y_n=\Sp_{2n}(\Z)\setminus \Sp_{2n}(\R)$ and shifts $s$ over the fundamental domain for each lattice, and the total expression is understood weakly, by pairing with Schwartz operators~\cite{Keyl_2016}. The right-hand side of this equality is given by the projector $\Pi_t$ onto the symmetric subspace of the tensor product. By the structure of $G_M$, the expression under the integral is given by 
\begin{align}
     \int_{\chi_n} G_{X}^{\otimes t} dX&= \int_{Y_n}\int_{\CP(\Lambda)}  \sum_{x\in \Z^t\otimes \Lambda} \lrc{ e^{\Phi_M(x)} }\lrq{e^{-i2\pi\omega(s,\sum_i x_i)}} D(x)\, ds d\mu_n(\Lambda),\nonumber\\
    &= \int_{Y_n} \sum_{x\in {\sf A}_{t-1}\otimes \Lambda} \lrc{ e^{\Phi_M(x)} } D(x) d\mu_n(\Lambda),
\end{align}
where $\Phi_M$ collects the stabilizer phases. For $t=2,3 $ the expression in the curly brackets becomes independent of the lattice basis $M$; the resulting integrand then becomes a Siegel transform.
\\
First-moment integrals of such transforms over the spaces of lattices with covolume 1 were studied in~\cite{Siegel_1945} and have been extended to higher moments by~\cite{Rogers1955}. Subsequently, the case of first or second moments over the space $Y_n$ of symplectic lattices has been treated in~\cite{Buser1994}, \cite{athreya2018discrepancygeneralsymplecticlattices}, \cite{Kelmer_2019}, and most recently in the work ~\cite{Kim_Holm} that uses an approach based on Hecke operators (cf. \cite[Lemma 7]{Kim_flags}) to give two explicit formulas for the integrals of rank-2 Siegel transforms over $Y_n$: Up to the conditions on the test functions, \cite[Thm. 1.4]{Kim_Holm} considers a general $f : \mathbb{R}^{2n} \times \mathbb{R}^{2n} \rightarrow \C$ and expresses the integral of the associated Siegel transform with the Radon transform; and \cite[Thm. 1.5]{Kim_Holm} considers the special case $f = f_1 \otimes f_2$ for $f_1, f_2 : \mathbb{R}^{2n} \rightarrow \C$ and utilizes the Fourier transform.

In this paper we give a self-contained derivation of a second-moment formula (Theorem \ref{thm:2nd_moment}) that lies somewhere between these two results (particularly, a Fourier-theoretic second-moment formula for test functions not necessarily of the form $f_1 \otimes f_2$). We apply this version of the formula to compute the third-moment defect of the GKP ensemble. Our proof uses the Hecke operator-based approach of \cite{Kim_Holm}, and essentially proceeds along the same path. At the same time, it differs by virtue of pursuing an alternative argument at one of the steps of the strategy, namely by using the distribution of conditional Lagrangian subspaces for $\F_p^{2n}$ in the $\F_p$-rank analysis of pairs $(x_1, x_2) \in \Z^{2n} \times \Z^{2n}$ (cf. \cite[Sect. 3]{Kim_Holm}). This, in effect, makes it more straightforward to work with the full Hecke average over all relevant cosets of \eqref{eq:Hecke_congruence_subgroup_yolo}, rather than a large subset of these cosets (cf. \cite[Sect. 5]{Kim_Holm}).

Our reason for giving a full proof of the symplectic second-moment formula is twofold: One reason is that the modifications discussed above are best explained in the context of a self-contained proof. The other reason is expository and aims to communicate and explain the mechanism of Hecke operators and their applications to mathematical physics. Hecke operators are ubiquitous in number theory and geometry and have also recently found further application in complexity theory~\cite{deBoerSIVP}.

We leave mathematical details of the Hecke operator-based approach to deriving moment formulae for GKP designs to later sections, but for now, we briefly sketch their interpretation in the context of GKP codes:
Let $\Lambda\in Y_n$ be a symplectic lattice. The associated code space has dimension $\sqrt{|\Lambda^{\perp}/\Lambda|}=1$. To realize a less trivial code space, one may scale by a factor $\sqrt{p}$ (for an integer $p > 1$) to obtain $\Lambda_p=\sqrt{p}\Lambda$ with $\Lambda_p^{\perp}/\Lambda_p \simeq \Lambda/p\Lambda$. For $p$ prime, this realizes a representation of $H(\F_p^n)$ on the associated code space, with underlying symplectic vector space $(\F_p^{2n}, \omega)$. In quantum information terms, we have obtained a collection of $n$ $p$-dimensional qudits. A Lagrangian subspace of $\F_p^{2n}$ determines a maximal commuting family of logical displacements, and fixing one of its $p^n$ stabilizer characters specifies a unique state projectively. In the quantum information literature, such states are known as stabilizer states.
${\rm Lag}(\Lambda/p\Lambda)$ is the collection of all Lagrangian subspaces that can be obtained from $\Lambda\in Y_n$. A Hecke operator $T_p$ averages over all the corresponding stabilizer lattices,
\begin{align}
    T_p F(\Lambda) &= \frac{1}{|{\rm Lag}(\Lambda/p\Lambda)|} \sum_{W\in {\rm Lag}(\Lambda/p\Lambda)} F(p^{-1/2} M_W),\nonumber\\
    M_W &= \lrc{x\in \Lambda:\, x \text{ (mod $p \Lambda$)} \in W }.
\end{align}
Under this average, $\Lambda_W=p^{-1/2} M_W$ again is a symplectic lattice. 
For the Siegel transforms considered below, we compute the pointwise limit of $T_p$ as $p\to\infty$ and identify it with the Haar mean.

\subsection*{Organization}
With this manuscript, we aim to provide a self-contained introduction into the topic of moment formulae for functions on the space of lattices and the application of Hecke operators in their derivation. This introduction is particularly aimed at non-experts in number theory and attempts to be self-contained and pedagogical. In this spirit, we spend section~\ref{sec:Hecke} on deriving $1^{\rm st}$- and $2^{\rm nd}$ moment formulae over the space of symplectic lattices $Y_n$ in a step-by-step fashion. 
Section~\ref{sec:designs} introduces (rigged) quantum $t$-designs, GKP distributions, explains the previous results of~\cite{Conrad_2025} and establishes Theorem~\ref{thm:main}. Readers familiar with the methods of the first part may safely skip to this section after peeking at the statement of Theorem~\ref{thm:2nd_moment}.

\section{Hecke operators and the Siegel measure}\label{sec:Hecke}

Let $n\geq 1$ and consider the symplectic form
\begin{equation}
    J=\begin{pmatrix} 0 & I_n \\ -I_n & 0 \end{pmatrix},
\end{equation}
and the groups
\begin{equation}
    G=\Sp_{2n}(\R),\quad \Gamma=\Sp_{2n}(\Z), \quad Y_n=\Gamma\setminus G.
\end{equation}
Since $\Gamma\subset G$ is a lattice, $Y_n$ carries a unique right-$G$-invariant probability measure $\mu_n$, so $\mu_n(Ah)=\mu_n(A)$ for $h\in G$.
A symplectic lattice is given by $\Lambda=\Z^{2n}g\subset\R^{2n}$ for $g\in G$. Equivalently, it is a full lattice satisfying $\omega(\Lambda,\Lambda)\subseteq\Z$ and ${\rm vol}(\Lambda)=1$. 
The group $G$ acts transitively on these lattices, with $\Gamma$ stabilizing the standard lattice $\Z^{2n}$.
\medskip \\
For a $2n$-dimensional symplectic vector space $(V, \omega)$, a Lagrangian subspace -- or polarization --  $L\subset V$ is an $n$-dimensional subspace such that $\omega(L, L)=0$, and we denote by $\Lag(V)$ the collection of all Lagrangian subspaces.
\medskip \\
For a prime $p$, consider $V=\F_p^{2n}$ and let $L_0=\F_p^n\oplus\{0\}^n$. Since $G_p=\Sp_{2n}(\F_p)$ acts transitively on $\Lag(V)$, every Lagrangian is $L=L_0g_p$ for some $g_p\in G_p$.
Because 
\begin{equation}
    (x,\,0)
    \begin{pmatrix}
        A & B \\ C & D 
    \end{pmatrix}
    = (xA,\, xB)
\end{equation}
we have that 
\begin{equation}{\rm Stab}_{G_p}(L_0)=
    \lrc{\begin{pmatrix} A & 0 \\ C & D\end{pmatrix}
    \in G_p},
\end{equation}
and each Lagrangian subspace is fully specified by an element of ${\rm Stab}_{G_p}(L_0)\setminus G_p$.
We set $a_p=p^{-\frac{1}{2}}(I_n\oplus pI_n)$ and define the Hecke congruence subgroup
\begin{equation}\label{eq:Hecke_congruence_subgroup_yolo}
 \Gamma_p =\Gamma \cap a_p^{-1}\Gamma a_p =\lrc{\begin{pmatrix} A & B \\ C & D\end{pmatrix}\in \Gamma, \, B \equiv 0 \text{ (mod $p$)}}. 
\end{equation}
Reduction modulo $p$ implements a surjection $\rho_p:\Gamma\twoheadrightarrow G_p$, for which the inverse image of ${\rm Stab}_{G_p}(L_0)$ is $\Gamma_p$, with
\begin{equation}
        \Gamma_p\backslash\Gamma\simeq\Lag(\F_p^{2n}),\qquad
    \Gamma_p\gamma\longmapsto L_0\rho_p(\gamma)
\end{equation}
%
%
%
For a module $X$ we write $X^{\perp}$ for the symplectic dual relative to $\omega$. Observe that for a symplectic lattice $\Lambda=\Lambda^{\perp}\in Y_n$, we have
\begin{equation}
    (p^{\frac{1}{2}}\Lambda)^{\perp}/(p^{\frac{1}{2}}\Lambda)\simeq  V_p\coloneqq \Lambda/p\Lambda \equiv \F_p^{2n},
\end{equation}
where the symplectic form of $\R^{2n}\supset \Lambda$ descends to the symplectic form on $\F_p^{2n}$. 
Let $W\in \Lag({V_p})$ and define 
\begin{equation}
    M_W =\{v\in \Lambda:\, \overline{v}\in W\},\quad \Lambda_W=p^{-\frac{1}{2}}M_W,
\end{equation}
where $\overline{v}\coloneqq v$ (mod $p\Lambda$). 
\begin{lemma}
    For every Lagrangian $W\subset V_p$ we have
\begin{equation}
    [\Lambda : M_W]=p^n,\quad M_W^{\perp}=p^{-1}M_W,\quad \Lambda_W^{\perp}=\Lambda_W.
\end{equation}
In particular, $\Lambda_W\in Y_n$.
\end{lemma}
\begin{proof}
    The first identity follows from $\Lambda/M_W \equiv V_p/W$, which has dimension $n$ over $\F_p$. Since $p\Lambda\subset M_W$ we also have $M_W^{\perp}\subset p^{-1}\Lambda$. Let $y\in \Lambda$, such that $p^{-1}y\in M_W^{\perp}$. This implies
    \begin{equation}
        p^{-1}y\in M_W^{\perp} \Leftrightarrow \omega(y, M_W)\in p\Z \Leftrightarrow \overline{y}\in W^{\perp}=W.
    \end{equation}
    Therefore $y \in M_W$, and $M_W^{\perp}=p^{-1}M_W$, which also implies $\Lambda_W^{\perp}=\Lambda_W$.
\end{proof}

Let $N_n(p)=|\Lag(\F_p^{2n})|$ and consider an $r$ dimensional sublagrangian $U\subseteq W$. We have $W/U \subseteq U^{\perp}/U$. Since $U$ is isotropic, $\dim (U^{\perp}/U)=\dim U^{\perp}-\dim U= 2n-r -r=2(n-r) $. Therefore 
\begin{equation}
    |\{W\in \Lag(V_p), U\subseteq W\}|=N_{n-r}(p).
\end{equation}
In particular we consider for $r=1$ a line $\ell \subset W \subseteq V_p$ and count the number of such pairs $(\ell, W)$.
There are $(p^{2n}-1)/(p-1)$ choices for $\ell$ in $V_p$ before fixing $W$, but once $W$ is fixed there remain only $(p^{n}-1)/(p-1)$ choices for $\ell$. Hence we have
\begin{equation}
    N_n(p)\frac{p^n-1}{p-1} = \frac{p^{2n}-1}{p-1}N_{n-1}(p) \Leftrightarrow \frac{N_n(p)}{N_{n-1}(p)}=p^n+1.
\end{equation}
Iterating this equation finally yields 
\begin{eqnarray}
    N_n(p)=\prod_{i=1}^n(p^i+1).
\end{eqnarray}
In particular, for an isotropic $U$ of dimension $1$ or $2$, the probability for a random $W\subset V_p$ to contain $U$ is
\begin{align}
    \alpha_1=\frac{N_{n-1}(p)}{N_n(p)}&=\frac{1}{p^n+1},\quad \alpha_2=\frac{1}{(p^n+1)(p^{n-1}+1)},\nonumber\\ &\alpha_1-\alpha_2 = p^{n-1}\alpha_2.
\end{align}

For a prime $p$, the (symplectic) Hecke operator is defined as 
\begin{equation}
    (T_p\Phi)(\Lambda) = \frac{1}{N_n(p)}\sum_{W\in \Lag(\Lambda/p\Lambda)}\Phi(\Lambda_W)
\end{equation}
and satisfies 
\begin{equation}
    \int_{Y_n} \Phi (\Lambda)d\mu_n(\Lambda) = \int_{Y_n} T_p\Phi (\Lambda)d\mu_n(\Lambda).
\end{equation}
To see this, consider the space $X_p = \Gamma_p \setminus G$. Each point in $X_p$ labels a pair $(\Lambda, W)$ of a lattice $\Lambda\in Y_n$ and Lagrangian subspace $W\subset \Lambda/p\Lambda$, given by the two projections
\begin{equation}
    \pi_1(\Gamma_p g) = \Gamma g,\quad \pi_2(\Gamma_p g) =\Gamma a_p g,
\end{equation}
equivalently
\begin{equation}
    \pi_1(\Lambda, W) = \Lambda,\quad \pi_2(\Lambda, W) =\Lambda_W.
\end{equation}
By the coset identification above, the fiber of $\pi_1$ over $\Lambda$ is
\begin{equation}
    \pi_1^{-1}(\Lambda) = \{(\Lambda, W):\; W\in \Lag(\Lambda/p\Lambda)\}
\end{equation}
and $T_p$ realizes exactly the averaging of $\pi_2$ over a fiber of $\pi_1$ and a uniform right $G$-invariant probability measure on $X_p$ is given by
\begin{equation}
    \int_{X_p} F(x) d\nu_p(x) = \int_{Y_n} \frac{1}{N_n(p)}\sum_{x\in \pi_1^{-1}(y)} F(x) d\mu_n(y).
\end{equation}
Set $F=\Phi\circ\pi_2$, then
\begin{equation}
    \int_{X_p} \Phi(\pi_2(x)) d\nu_p(x) = \int_{Y_n} \underbrace{\frac{1}{N_n(p)}\sum_{x\in \pi_1^{-1}(y)}\Phi(\pi_2(x))}_{(T_p\Phi)(y)}d\mu_n(y).
\end{equation}
The measure $\nu_p$ is right $G-$invariant. Let $h\in G$ and let $\Phi$ be a bounded and measurable function. Then using right $G$-equivariance of $\pi_2$ we obtain
\begin{equation}
    \int_{Y_n}\Phi(yh) d((\pi_2)_*\nu_p)(y) = \int_{X_p}\Phi (\pi_2(x)h)d\nu_p(x)=\int_{X_p}\Phi(\pi_2(x))d\nu_p(x),
\end{equation}
Uniqueness of the invariant probability measure gives $(\pi_2)_*\nu_p=\mu_n$. For integrable $\Phi$ this yields for any prime $p>1$
\begin{equation}
\int_{Y_n} T_p\Phi d\mu_n = \int_{X_p} \Phi \circ \pi_2 d\nu_p = \int_{Y_n} \Phi d\mu_n,
\end{equation}
So the uniform measure $\mu_n$ on $Y_n$ is invariant under pushforward by $T_p$. This allows us to evaluate the integrals considered below.


\subsection{Symplectic first moments}
Let $f\in C(\R^{2n})$ be a continuous function with $|f(x)|=O((1+\|x\|)^{-2n-\epsilon})$, $\epsilon>0$. Its Siegel transform is given by
\begin{equation}
    S_f(\Lambda)=\sum_{\lambda \in \Lambda^{\times}}f(\lambda).
\end{equation}
where we used the shorthand $S^{\times} = S\setminus\{0\}$ to denote the exclusion of $\{0\}$ for any set $S$ and the decay assumption guarantees absolute convergence of this series.

\begin{theorem}[Symplectic Siegel mean value,~\cite{Buser1994}]\label{thm:1st_moment}
    Let $f\in C(\R^{2n})$ be a continuous function with $|f(x)|=O((1+\|x\|)^{-2n-\epsilon})$, $\epsilon>0$, then
    \begin{equation}
        \int_{Y_n} S_f(\Lambda)d\mu_n(\Lambda) = \int_{\R^{2n}} f(x) dx.
    \end{equation}
\end{theorem}
\begin{proof}
First assume $f\geq0$. The previous discussion showed that the integral is invariant under application of the Hecke operator
    \begin{equation}
        \int_{Y_n} S_f(\Lambda)d\mu_n(\Lambda) = \int_{Y_n} T_pS_f(\Lambda)d\mu_n(\Lambda).
    \end{equation}
    We have
    \begin{align}
        T_pS_f(\Lambda) &= \frac{1}{N_n(p)} \sum_{W\in \Lag(\Lambda/p\Lambda)} \sum_{v\in \Lambda_W^{\times}}f(v)
    \end{align}
    Since we have 
    \begin{equation}
        \Lambda_W^{\times} = p^{-\frac{1}{2}} \{v\in \Lambda : \overline{v}\in W^{\times}\} \cup p^{\frac{1}{2}}\Lambda^{\times}
    \end{equation}
    and for any fixed $v\in\Lambda\setminus p\Lambda$ the probability of $\overline{v}\in W^{\times}$ is $\alpha_1(p)=1/(p^n+1)$ as computed earlier, we find
    \begin{equation}
        T_pS_f(\Lambda) = \frac{1}{p^n+1} \sum_{v\in \Lambda^{\times}} f(p^{-\frac{1}{2}}v) + \frac{p^n}{p^n+1} \sum_{u\in \Lambda^{\times}} f(p^{\frac{1}{2}}u).
    \end{equation}
    The second sum is $O(C_{\Lambda,f}p^{-n-\epsilon/2})$, since $\sum_{u\in\Lambda^{\times}}\|u\|^{-2n-\epsilon}<\infty$. The first term yields the Riemann integral 
    \begin{equation}
        \lim_{p\to\infty} T_pS_f(\Lambda) = \int_{\R^{2n}} f(x)dx\eqqcolon c_f.
    \end{equation}

    To pass from pointwise convergence $T_pS_f\to c_f$ to equality of the integrals, let $K>0$ and $F_K={\rm min}(S_f, K)$. $T_p$ is an averaging operator, hence $T_p\geq 0$ and $T_p1=1$. Therefore we also always have $T_pF_K\leq \min(T_pS_f, K)$. Invariance and bounded convergence yield
   \begin{equation}
        \int_{Y_n}F_K\,d\mu_n
        =\int_{Y_n}T_pF_K\,d\mu_n
        \leq\int_{Y_n}\min(T_pS_f,K)\,d\mu_n
        \xrightarrow{p\to\infty}\min(c_f,K).
    \end{equation}
    In the limit $K\to\infty$, we get $\int_{Y_n}S_f(\Lambda)d\mu_n(\Lambda) \leq c_f<\infty$. In particular, each $T_pS_f$ is integrable and positive per assumption. Hence, we can apply Fatou's lemma so to obtain
    \begin{equation}
        c_f=\int_{Y_n}\lim_{p\to\infty} T_pS_f(\Lambda) d\mu_n(\Lambda) \leq \liminf_{p\to\infty} \int_{Y_n} T_p S_f(\Lambda)d\mu_n(\Lambda).
    \end{equation}
    The right-hand side equals $\int_{Y_n}S_f\,d\mu_n$, hence equality holds. Applying this result to $|f|$ proves absolute integrability and finally splitting the real and imaginary parts of $f$ into their positive and negative parts proves the general case.
\end{proof}

\subsection{Symplectic second moments}
We are also interested in the equivalent expressions for continuous two-variable functions $H:\, \R^{2n}\times \R^{2n}\to \C$. If $H\in C(\R^{2n}\times \R^{2n})$ is such that, for $\epsilon>0$, $|H(x, y)|= O((1+\|x\|+\|y\|)^{-4n-\epsilon})$ the generalized Siegel transform
\begin{equation}
    \Theta_H^{\times}(\Lambda)=\sum_{v,w \in \Lambda^{\times}}H(v,w)
\end{equation}
converges absolutely. We define for a function $G\in C(\R^{2n}\times \R^{2n})$ with $|G(x, y)|= O((1+\|x\|+\|y\|)^{-4n-\epsilon})$ the series of integrals
\begin{equation}
    \CL(G)\coloneqq \frac{1}{2\zeta(2n)}\sum_{a, b\in \Z^{\times}} \int_{\R^{2n}} G(az, bz) dz,\label{eq:def_CL}
\end{equation}
as well as the partial Fourier transform
\begin{equation}
    (\CF_{2}H)(x, \eta) = \int_{\R^{2n}} H(x,y) e^{-i2\pi \omega(y, \eta)}dy.
\end{equation}

\begin{theorem}\label{thm:2nd_moment}
    Let $n\geq 2$ and $H\in C(\R^{2n}\times \R^{2n})$ be a continuous function with \begin{equation}|H(x, y)|+|\CF_2H(x, y)|\leq  C(1+\|x\|+\|y\|)^{-4n-\epsilon}\label{eq:bound_H}
    \end{equation} for $\epsilon >0$ and some constant $C>0$.
    We have
    \begin{equation}
        \int_{Y_n}\Theta_H^{\times}(\Lambda) d\mu_{n}(\Lambda) = \int_{\R^{2n}\times \R^{2n}} H(x,y) dxdy + \CL(H) + \CL(\CF_2H).
    \end{equation}
\end{theorem}
\begin{proof}
    The proof proceeds similarly to the previous one with more terms to keep track of. Let $H_p[v,w]=H(v/\sqrt{p}, w/\sqrt{p})$ and define
    \begin{align}
        A_p(H)&=\sum_{v,w\in\Lambda^{\times}}  H_p[v,w],\quad 
        &&C_p(H)=\sum_{\substack{v,w\in\Lambda^{\times} \\ p|\omega(v,w)}} H_p[v,w], \nonumber \\ 
        R_p(H)&=\sum_{\substack{v,w\in\Lambda^{\times} \\\dim_{\F_p}\langle \overline{v}, \overline{w}\rangle=1}} H_p[v,w], \quad
        &&Z_p(H)=\sum_{v,w\in p\Lambda^{\times}}H_p[v,w], \nonumber\\
        Q_p(H) &= \sum_{\substack{v,w\in\Lambda^{\times} \\ \dim_{\F_p}\langle \overline{v}, \overline{w}\rangle=2\\ p|\omega(v,w)}} H_p[v,w],
    \end{align}
    where $C_p(H) = Q_p(H)+R_p(H)+Z_p(H)$.
    For $a, b \in V_p$, let $\beta_p(a,b)$ be the probability that a random Lagrangian contains both vectors. Then
    \begin{equation}
        \beta_p(a,b)=\begin{cases}
            1,& a=b=0 \\ \alpha_1,& \dim{\rm span}(a,b) =1, \\ \alpha_2,& \dim{\rm span}(a,b) =2\, \wedge\, \omega(a,b)=0 \\ 0 &  \dim{\rm span}(a,b) =2\,  \wedge\, \omega(a,b)\neq 0
        \end{cases}
    \end{equation}
    and, as before, 
    \begin{align}
        T_p\Theta_H^{\times}(\Lambda) &= \sum_{v,w \in \Lambda^{\times}} \beta_p(\overline{v},\overline{w})H_p[v,w] \nonumber\\ 
        &= \alpha_2Q_p(H) + \alpha_1 R_p(H)+ Z_p(H) \nonumber\\
        &= \alpha_2C_p(H) + (\alpha_1-\alpha_2) R_p(H)+ (1-\alpha_2)Z_p(H).
    \end{align}
    In the next step we expand
    \begin{equation}
        1_{p|\omega(v,w)}=\frac{1}{p}\sum_{t\in \F_p} e^{i2\pi t \omega(v,w)/p}
    \end{equation}
    such that
    \begin{align}
        C_p(H) &= \frac{1}{p}\sum_{v\in\Lambda^{\times}}\sum_{t\in \F_p}  \lr{ \sum_{w\in \Lambda} H_p[v,w]e^{i2\pi t \omega(v,w)/p} - H_p[v,0]}\nonumber\\ 
        &= \frac{1}{p}A_p(H)+ \frac{1}{p}\sum_{t\in \F_p^{\times}}\sum_{v\in\Lambda^{\times}} \lr{ \sum_{w\in \Lambda} H_p[v,w]e^{i2\pi t \omega(v,w)/p}- H_p[v,0]}.\label{eq:40}
    \end{align}
    Define further \begin{equation}
        X_p(H)= \sum_{v\in \Lambda \setminus p\Lambda} H_p[v,0],\quad Y_p(H)=\sum_{\substack{v\in p\Lambda^{\times}\\ w\in \Lambda^{\times}}}H_p[v,w].
    \end{equation}
    The last sum in eq.~\eqref{eq:40} splits as
    \begin{align}
        &\hspace{1cm}\frac{1}{p}\sum_{t\in \F_p^{\times}}\sum_{v\in\Lambda^{\times}} \lr{ \sum_{w\in \Lambda} H_p[v,w]e^{i2\pi t \omega(v,w)/p}- H_p[v,0]}\nonumber\\ &= \frac{1}{p}\sum_{t\in \F_p^{\times}}\lrc{\sum_{v\in \Lambda \setminus p\Lambda}+\sum_{v\in p\Lambda^{\times}}} \lr{ \sum_{w\in \Lambda} H_p[v,w]e^{i2\pi t \omega(v,w)/p}- H_p[v,0]}\nonumber\\
        &=\frac{1}{p}\sum_{t\in \F_p^{\times}}\sum_{v\in \Lambda \setminus p\Lambda}\lr{ \sum_{w\in \Lambda} H_p[v,w]e^{i2\pi t \omega(v,w)/p}- H_p[v,0]} + \frac{p-1}{p}Y_p(H)\nonumber\\
        &=\frac{1}{p}\sum_{t\in \F_p^{\times}}\sum_{v\in \Lambda \setminus p\Lambda}\lr{ \sum_{w\in \Lambda} H_p[v,w]e^{i2\pi t \omega(v,w)/p}} + \frac{p-1}{p}(Y_p(H)-X_p(H))\label{eq:42}
    \end{align}

    To simplify the remaining expression we use Poisson summation,
    \begin{equation}
        \sum_{w\in \Lambda} H(v/\sqrt{p},w/\sqrt{p})e^{i2\pi t \omega(v,w)/p} = p^n\sum_{u\in\Lambda}\CF_2H\lr{\frac{v}{\sqrt{p}}, \frac{pu+tv}{\sqrt{p}}}.\label{eq:Poisson}
    \end{equation}
    For $z=pu+tv$ we have $\overline{z}=t\overline{v}$. Hence each pair of nonzero collinear $\overline{z}, \overline{v}$ uniquely determines a $t\in \F_p^{\times}$. Once a lift of $t$ is fixed, the vector $u=(z-tv)/p$ is also uniquely specified.
    Therefore, inserting eq.~\eqref{eq:Poisson} into eq.~\eqref{eq:42}, we get
    \begin{align}
        \frac{1}{p}\sum_{\substack{v\in \Lambda\setminus p\Lambda, w\in \Lambda\\ t\in \F^{\times}_p}} H_p[v,w]e^{i2\pi t \omega(v,w)/p} &= p^{n-1}\sum_{\substack{v\in \Lambda\setminus p\Lambda, u\in \Lambda\\ t\in \F^{\times}_p}} \CF_2H\lr{\frac{v}{\sqrt{p}}, \frac{pu+tv}{\sqrt{p}}} \nonumber\\
        &=p^{n-1}\sum_{\substack{v, z \in \Lambda\setminus p\Lambda, \\ \dim_{\F_p}\langle\overline{v}, \overline{z} \rangle=1}} \CF_2H\lr{\frac{v}{\sqrt{p}}, \frac{z}{\sqrt{p}}}
    \end{align}
    Calling \begin{equation}
        R^{nn}_p(\CF_2H) = \sum_{\substack{v, z \in \Lambda\setminus p\Lambda, \\ \dim_{\F_p}\langle\overline{v}, \overline{z} \rangle=1}} \CF_2H\lr{\frac{v}{\sqrt{p}}, \frac{z}{\sqrt{p}}},
    \end{equation}
    we have now obtained
    \begin{equation}
        C_p(H)=\frac{1}{p}A_p(H)+p^{n-1}R^{nn}_p(\CF_2H) + \frac{p-1}{p}(Y_p(H)-X_p(H)),
    \end{equation}
    and, in total,\begin{align}
    T_p\Theta^{\times}_H &= \frac{\alpha_2}{p}A_p(H)+ p^{n-1}\alpha_2 \lrq{R_p(H) + R_p(\CF_2H)} +E_p(H), \nonumber\\
    E_p(H) &= \alpha_2\frac{p-1}{p}\lrq{Y_p(H)-X_p(H)} + (1-\alpha_2)Z_p(H) \nonumber \\
    &\hspace{2cm} -p^{n-1}\alpha_2\lrq{R_p(\CF_2H)-R_p^{nn}(\CF_2H)}.\label{eq:Theta_total}
    \end{align}
    It is time to consider the limit $p\to \infty$. We fix $s=4n+\epsilon$. A key observation we use is that for $\Lambda\subset \R^{2n}$, $A\geq 1$ and $s>2n$, there exist constants $C_{L, s}$, such that
    \begin{equation}
        \sum_{v\in \Lambda} \lr{A+\frac{\|v\|}{\sqrt{p}}}^{-s} \leq C_{L, s} p^{n}A^{2n-s}.\label{eq:bound_tool}
    \end{equation}
    To verify this, set $R=A\sqrt{p}$, such that
    \begin{equation}
    \sum_{v\in \Lambda} \lr{A+\frac{\|v\|}{\sqrt{p}}}^{-s}= A^{-s}\sum_{v\in \Lambda} \lr{1+\frac{\|v\|}{R}}^{-s}.\label{eq:sum_bound}
    \end{equation}
    Let $B_R\subset \R^{2n}$ be the ball of radius $R$ and define the lattice shells
    \begin{equation}
        S_j=\{v\in\Lambda: 2^jR < \|v\|\leq 2^{j+1}R\},\quad j\geq 0.
    \end{equation}
    Inside $B_R$, every summand of eq.~\eqref{eq:sum_bound} is at most one and the corresponding partial sum can be bounded by $C_{L}R^{2n}$ points within the $B_R$. On each shell $S_j$, the contribution to the sum is proportional to $2^{-js}$, while $|S_j|\leq C'_L 2^{2n(j+1)}R^{2n}$. Consequently,
    \begin{equation}
        \sum_{v\in \Lambda} \lr{1+\frac{\|v\|}{R}}^{-s} \leq C_L R^{2n}+C'_L(2R)^{2n}\sum_{j\geq 0} 2^{(2n-s)j}.
    \end{equation}
For $s>2n$ the geometric sum converges absolutely, and we arrive by resubstituting $R$ at 
\begin{equation}
    \sum_{v\in \Lambda} \lr{A+\frac{\|v\|}{\sqrt{p}}}^{-s}\leq C_{L,s} p^n A^{2n-s}\label{eq:bound_sumA}
\end{equation}
for some constant $C_{L,s}\geq 0$. 
Using eq.~\eqref{eq:bound_tool} and $A=1$, $s>2n$ we can thus bound
\begin{align}
    |X_p(H)|&\leq C \sum_{v\in\Lambda} (1+\|v\|/\sqrt{p})^{-s} = O(p^n),
\end{align}
Using $A=1+\sqrt{p}\|v\|$ we next have
\begin{align}
    |Y_p(H)|&\leq C \sum_{v, w\in\Lambda^{\times}} (1+\sqrt{p}\|v\|+\|w\|/\sqrt{p})^{-s}\nonumber\\
    &\leq C' p^n \sum_{v\in \Lambda^{\times}} (1+\sqrt{p}\|v\|)^{2n-s} \nonumber\\
    &\leq p^{2n-s/2}\sum_{v\in \Lambda^{\times}} \|v\|^{2n-s} \leq C'' p^{2n-s/2}=C'' p^{-\epsilon/2}, 
\end{align}
where the final series converges since $s-2n>2n$.

For $Z_p(H)$ we write $(v,w)=(pa, pb)\in p\Lambda^{\times}$. Then
\begin{equation}
    |Z_p(H)|\leq Cp^{-s/2}\sum_{a, b\in \Lambda^{\times}} (p^{-\frac{1}{2}} +\|a\|+\|b\|)^{-s},
\end{equation}
where the series converges for $\epsilon >0$.

Finally, the last term $B_p(\CF_2H)=R_p(\CF_2H)-R_p^{nn}(\CF_2H)$ consists of pairs of vectors $(v,w)$ with either $\overline{v}=0\wedge \overline{w}\neq 0 $ or $\overline{v}\neq 0$ and $\overline{w}=0$. Their absolute sums are bounded by the $Y_p$ estimate applied to $(\CF_2H)(x,y)$ and $(\CF_2H)(y,x)$, respectively. Hence there is a new constant $C'''$, such that
\begin{equation}
    |B_p(\CF_2H)| \leq C'''p^{-\epsilon/2}.
\end{equation} 

Using $\alpha_2 = O(p^{1-2n})$ and $p^{n-1}\alpha_2=O(p^{-n})$ the contributions to the terms in $E_p$ are hence as follows.
\begin{table}[ht]
    \centering
    \begin{tabular}{lll}
        \hline
        Term & Bound under eq.~\eqref{eq:bound_H} & Contribution to $E_p$ \\
        \hline
        $X_p(H)$
        & $O(p^n)$
        & $O(p^{1-n})$ \\

        $Y_p(H)$
        & $O(p^{-\varepsilon/2})$
        & $O(p^{1-2n-\varepsilon/2})$ \\

        $Z_p(H)$
        & $O(p^{-2n-\varepsilon/2})$
        & $O(p^{-2n-\varepsilon/2})$ \\

        $R_p(\CF_2H)-R_p^{\mathrm{nn}}(\CF_2H)$
        & $O(p^{-\varepsilon/2})$
        & $O(p^{-n-\varepsilon/2})$ \\
        \hline
    \end{tabular}
    \caption{Bounds for a fixed lattice, decay exponent $4n+\varepsilon$, where $\varepsilon>0$.}
    \label{tab:error-bounds}
    \end{table}

In particular, $E_p(H)=O(p^{1-n})$. Now that $E_p(H)$ is out of the way, we look again at eq.~\eqref{eq:Theta_total}. For large $p$, $\alpha_2/p\sim p^{-2n}$, such that
\begin{align}
    \lim_{p\to \infty} \frac{\alpha_2}{p}A_p(H)=\lim_{p\to \infty} \frac{\alpha_2}{p}\sum_{v,w \in \Lambda^{\times}} H(v/\sqrt{p}, w/\sqrt{p}) =\int_{\R^{2n}\times \R^{2n}} H(x,y) dxdy,
\end{align}
as before, realizes a Riemann integral. 

To take the limit for $p^{n-1}\alpha_2R_p(H)$, we split the sum into pairs collinear over $\R$ and pairs independent over $\R$.
Every pair of nonzero collinear lattice vectors can be represented by $(v,w)=(az, bz), z\in \Lambda^{\times}$ where $a>0, b\neq 0$ and ${\rm gcd}(a,|b|)=1$. Since $a,b$ cannot be simultaneously divisible by $p$, we have $\dim_{\F_p}\langle \overline{v},\overline{w}\rangle=1$ precisely when $z\not\in p\Lambda$. %
The splitting is
\begin{equation}
    R_p(H)=\frac{1}{2}\sum_{\substack{a, b\in \Z^{\times}\\ {\rm gcd}(a, |b|)=1}}\sum_{z\in \Lambda\setminus p\Lambda} H(az/\sqrt{p},bz/\sqrt{p})+I_p(H), \label{eq:R_psplit}
\end{equation}
where the factor $1/2$ accounts for the overcounting of pairs $(v,w)=(az, bz)$ from triples $(a, b, z)\sim (-a,-b,-z) $ that yield the same initial pair, and

\begin{equation}
    I_p(H)=\sum_{\substack{v,w\in \Lambda^{\times} \\ \dim_{\R}\langle v,w \rangle =2 \\ \dim_{\F_p}\langle\overline{v}, \overline{w}\rangle=1 }}H(v/\sqrt{p}, w/\sqrt{p}).
\end{equation}

We can  bound $I_p(H)$ by the same argument as in eq.~\eqref{eq:bound_tool}. To this end we first count the  number of pairs $N_p(R)$ of vectors contributing to the sum, i.e., $v,w\in \Lambda^{\times}$ with $\dim_{\F_p}\langle \overline{v}, \overline{w}\rangle=1,\dim_{\R}\langle v, w\rangle=2 $ and $\|v\|+\|w\|\leq R$. 

To make this argument more intuitive to follow, we will deviate from our previous convention and use column vectors here.
Let $B$ be a basis for $\Lambda=B\Z^{2n}$ where $v=Bm, w=Bk$. The assumptions imply that for each $i=1,\hdots,2n$
\begin{equation}
    \sqrt{m_i^2+k_i^2} \leq \|m\|+\|k\| \leq \|B^{-1}\|R.\label{eq:mk_bound}
\end{equation} 
Since 
$\dim_{\F_p}\langle \overline{v}, \overline{w}\rangle=1$ we also have $\dim_{\F_p}\langle \overline{k}, \overline{m} \rangle=1$, where here we understand $\overline{k}$ as the reduction $k$ (mod $p$). Consider now the lattice on the space of coefficients $B_{km}=(k,m)\in \Z^{2n\times 2}$, where $B_{km}$ has rows $r_i=(k_i, m_i)\in \Z^2$. 
The number of pairs of points $v, w$ contributing to the sum in $I_p(H)$, with $\|v\|+\|w\|\leq R$, is the same as the number of choices of rows $\{r_i\}$ for $B_{km}$ that satisfy the rank and size conditions resulting from the requirements on $v, w$.

The matrix $\overline{B_{km}}$ has rank 1 over $\F_p$ and since row-rank equals column rank, we have that its reduced rows span a single line
\begin{equation}
    \ell={\rm span}_{\F_p}\{\overline{r}_1,\hdots,\overline{r}_{2n}\}\subset \F_p^2
\end{equation}
and every $r_i$ belongs to
the planar lattice
\begin{equation}
    M_{\ell} = \lrc{r\in \Z^2: \overline{r}\in \ell},
\end{equation}
which spans $[\Z^2 : M_{\ell}]=p$ integer cosets, as verified via
\begin{equation}
    p^2=[\Z^2:p\Z^2] = [\Z^2:M_{\ell}][M_{\ell}:p\Z^2]=[\Z^2:M_{\ell}]p.
\end{equation}
By eq.~\eqref{eq:mk_bound}, every row $r_i$ of $B_{km}$ lies in $M_\ell\cap B_{L}$, where $L=C_\Lambda R$ and $C_\Lambda=\|B^{-1}\|$. 
Since $B_{km}$ has real rank two, only lines for which $M_\ell\cap B_L$ contains two independent vectors $u,v$ contribute. The centered fundamental parallelogram of $\Z u+\Z v\subseteq M_\ell$ lies in $B_L$, so the covering radius of $M_\ell$ is at most $L$. Its Voronoi cells centered in $B_L$ thus lie in $B_{2L}$ and each has area $p$. Each row has at most $|M_\ell\cap B_L|$ choices, and there are at most $p+1$ eligible lines. Therefore
\begin{equation}
    |M_{\ell}\cap B_L|\leq \frac{4\pi L^2}{p},\quad N_p(R)\leq C'_{\Lambda}(p+1)\lr{\frac{R^2}{p}}^{2n}
\end{equation}
for a new constant $C'_{\Lambda}$. We depict an example construction on $M_{\ell}$ and the Voronoi-shell structure in fig.~\ref{fig:ip-bound}
\begin{figure}[htbp]
    \centering
    \setlength{\panelheight}{0.40\textwidth}%
    \begin{subfigure}[t]{0.48\textwidth}
        \centering
        \includegraphics[height=\panelheight]{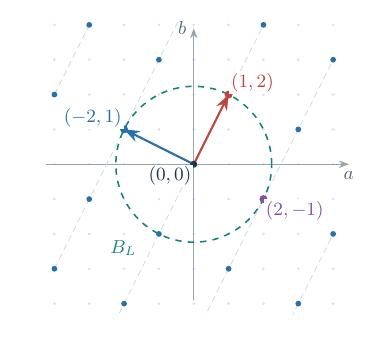}
    \end{subfigure}
    \hfill
    \begin{subfigure}[t]{0.48\textwidth}
        \centering
        \includegraphics[height=\panelheight]{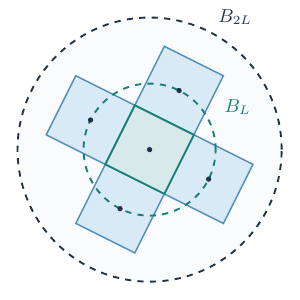}
    \end{subfigure}
    \caption{Geometry behind the bound on $I_p(H)$. The left illustration depicts the lattice $M_{\ell}$ at $\F_5$ for $v=(1, -2, 2, 0 )$ and $w=(2, 1, -1, 0)$; the right illustration depicts Voronoi cells for $M_{\ell}$ in the first shells of $B_L$.}
    \label{fig:ip-bound}
\end{figure}

Finally, let $s=4n+\epsilon$ denote the decay exponent in eq.~\eqref{eq:bound_H}, and let $S_I\subset \Lambda^{\times}\times \Lambda^{\times}$ denote the set of pairs of vectors $(v,w)$ occurring in the sum defining $I_p(H)$.
We split this sum into the inner region $\|v\|+\|w\|\leq\sqrt p$ and, for $j\geq0$, into the shells
\begin{equation}
    S_j=\lrc{(v,w)\in S_I: 2^j\sqrt p<\|v\|+\|w\|\leq2^{j+1}\sqrt p}.
\end{equation}
This now allows us to bound $I_p(H)$ using the same argument as we used to derive eq.~\eqref{eq:bound_sumA}. In the inner region the summand is at most $1$ and there are $N_p(\sqrt{p})$ pairs, while on $S_j$ the summand is bounded by $2^{-js}$ while there are $|S_j|\leq N_p(2^{j+1}\sqrt{p})$ pairs that contribute. Hence, with $D_{\Lambda}, D_{\Lambda}'$ constants independent of $p$,
\begin{align}
    |I_p(H)|&\leq C\sum_{(v,w)\in S_I} \lr{1+\frac{\|v\|+\|w\|}{\sqrt{p}}}^{-s} \nonumber\\ 
    &\leq CN_p(\sqrt{p})+C\sum_{j\geq 0} 2^{-js}N_p(2^{j+1}\sqrt{p})\nonumber\\
    &\leq D_{\Lambda}(p+1) + D_{\Lambda}'\sum_{j\geq 0}2^{-js}(p+1)2^{4n(j+1)}\nonumber\\
    &=   D_{\Lambda}(p+1) + 2^{4n}(p+1)D_{\Lambda}'\sum_{j\geq 0}2^{-j\epsilon}=O(p).
\end{align}
Finally, as $p^{n-1}\alpha_2\leq p^{-n}$, we obtain
\begin{equation}
    p^{n-1}\alpha_2|I_p(H)| = O(p^{1-n}),
\end{equation}
which vanishes in the limit for $n\geq 2$.

Now that also $I_p(H)$ is out of the way, we consider the remaining term in $R_p(H)$ in eq.~\eqref{eq:R_psplit}. Again we use that $p^{n-1}\alpha_2\sim p^{-n}$, and
for each fixed $(a,b)$ Riemann integration gives 

\begin{equation}
    \lim_{p\to\infty} p^{-n} \sum_{z\in \Lambda \setminus p\Lambda} H(az/\sqrt{p},\,bz/\sqrt{p}) = \int_{\R^{2n}} H(az, bz)dz,\label{eq:collinear}
\end{equation}
once we establish that the limit $p\to \infty$ can be exchanged with the sum over $(a,b)$ and that the same sum over $p\Lambda^\times$ vanishes in the limit. We use the decay condition to bound,
\begin{align}
    &p^{-n}\sum_{z\in \Lambda\setminus p\Lambda}|H(az/\sqrt p,bz/\sqrt p)| \nonumber \\
    &\quad\leq C p^{-n} (|a|+|b|)^{-4n-\epsilon}\sum_{z\in \Lambda^\times}\lr{\frac{1}{|a|+|b|}+\frac{\|z\|}{\sqrt{p}} }^{-4n-\epsilon}\nonumber \\
    &\quad\leq C_{\Lambda, s} (|a|+|b|)^{-2n}
\end{align}
where we use eq.~\eqref{eq:bound_tool} with the origin omitted. In this form, the same estimate holds for every $A>0$.
For $n>1$ the resulting sum over $a,b\in \Z^{\times}$ is finite and we are allowed to change the sum and limit by dominated convergence. For the second point, we have, in the limit $p\to\infty$,
\begin{align}
    p^{-n}\sum_{z\in p\Lambda^\times}|H(az/\sqrt p,bz/\sqrt p)|
    &\leq C p^{-n}\sum_{z\in \Lambda^\times} (1+\sqrt{p}(|a|+|b|)\|z\|)^{-4n-\epsilon} \nonumber \\
    &\leq C p^{-3n-\epsilon/2}\sum_{z\in \Lambda^\times}((|a|+|b|)\|z\|)^{-4n-\epsilon}.
\end{align}
The series in the last expression converges for $n>0$ under both the $z\in \Lambda^\times$ and $a,b \in \Z^{\times}$ sums and vanishes as $p\to\infty$. The omitted origin contributes $p^{-n}H(0,0)\to0$ for each fixed $(a,b)$.

The last ingredient we need is the Möbius function $\mu: \Z_{\geq 1} \to \{-1,0,1\}$, which has the properties
\begin{equation}
    \sum_{d|n}\mu(d) = \delta_{n,1},\quad \sum_{d\geq 1} \frac{\mu(d)}{d^s} = \frac{1}{\zeta(s)}
\end{equation} 
for $\Re s >1$. For a sufficiently well decaying sequence $T_{a,b}$ we can then write
\begin{align}
    \sum_{a,b\in \Z^{\times}}1_{{\rm gcd}(a,b)=1}T_{a,b}=\sum_{a,b\in \Z^{\times}} \sum_{d|{\rm gcd}(a,b)}\mu(d)T_{a,b}=\sum_{a,b\in \Z^{\times}} \sum_{\substack{d|a\\ d|b}}\mu(d)T_{a,b},
\end{align}
since the common divisors of $a, b$ are exactly the divisors of ${\rm gcd }(a,b)$.

So we can assemble
\begin{align}
    \lim_{p\to\infty} p^{-n}R_p(H) &= \frac{1}{2}\sum_{\substack{a,b\in \Z^{\times} \\ {\rm gcd(a, b)=1}}} \int_{\R^{2n}} H(az, bz)dz \nonumber\\
    &=\frac{1}{2}\sum_{a, b\in \Z^{\times}} \sum_{d\geq 1}\mu(d) 1_{(d|a)\wedge (d|b)} \int_{\R^{2n}} H(az, bz)dz \nonumber\\
    &=\frac{1}{2}\sum_{d\geq 1}\mu(d)\sum_{a, b\in \Z^{\times}}  \int_{\R^{2n}} H(daz, dbz)dz \nonumber\\
    &=\frac{1}{2}\sum_{d\geq 1}\frac{\mu(d)}{d^{2n}}\sum_{a, b\in \Z^{\times}}  \int_{\R^{2n}} H(az, bz)dz \nonumber\\ 
    &=\frac{1}{2\zeta(2n)}\sum_{a, b\in \Z^{\times}}  \int_{\R^{2n}} H(az, bz)dz.
\end{align}
With this expression, the definition for $\CL(\cdot)$ in eq.~\eqref{eq:def_CL} and repeating the same line of argumentation using $\CF_2H$ as argument of $R_p(\cdot)$, we can now assemble the $p\to \infty$ limit of eq.~\eqref{eq:Theta_total} to 
\begin{equation}
    \lim_{p\to\infty} T_p\Theta_H^{\times} = \int_{\R^{2n}\times \R^{2n}} H(x, y) dx dy + \CL(H)+\CL(\CF_2 H).\label{eq:74}
\end{equation}

For nonnegative $H$ satisfying the hypotheses, the same truncation and Fatou's lemma steps from the first-moment proof identify its integral with the finite limit in eq.~\eqref{eq:74}. For general $H$, we first prove integrability using a Gaussian envelope.

Let $G_R(x,y)=e^{-\pi(\|x\|^2+\|y\|^2)/R^2}$.
For $2^{j-1}< r=\|x\|+\|y\| \leq 2^j$ the pointwise bound is
\begin{equation}
    |H(x,y)|\leq C(1+r)^{-s}\leq C(1+2^{j-1})^{-s}< C2^{-s(j-1)}
\end{equation}
while $G_{2^j}(x,y)\geq e^{-\pi } $. So there is a new constant $C$, such that
\begin{equation}
    |H(x,y)| \leq C\sum_{j\geq 0} 2^{-sj}G_{2^j}(x,y).
\end{equation}
The partial Fourier transform is
\begin{equation}
 (\CF_2G_R)(x,\eta)=R^{2n}e^{-\pi\|x\|^2/R^2-\pi R^2\|\eta\|^2},
\end{equation}
so the nonnegative case and eq.~\eqref{eq:74} give, for $R\geq1$,
\begin{align}
 \int_{Y_n}\Theta_{G_R}^{\times}\,d\mu_n
 &=R^{4n}+\frac{R^{2n}}{2\zeta(2n)}
 \sum_{a,b\in\Z^\times}\left[(a^2+b^2)^{-n}+(R^{-2}a^2+R^2b^2)^{-n}\right]\nonumber\\
 &\leq(1+c_n)R^{4n},\qquad
 c_n=\frac{1}{\zeta(2n)}\sum_{a,b\in\Z^\times}(a^2+b^2)^{-n}<\infty.
\end{align}
Here we used $R^{-2}a^2+R^2b^2\geq R^{-2}(a^2+b^2)$, and finally obtain
\begin{equation}
 \int_{Y_n}\Theta_{|H|}^{\times}\,d\mu_n
 \leq C'(1+c_n)\sum_{j\geq0}2^{-(s-4n)j}<\infty,
 \qquad s=4n+\epsilon.
\end{equation}

Write $F=\Theta_H^\times\in L^1(\mu_n)$. Positivity, $T_p1=1$, and stationarity give
\begin{equation}
 \sup_p\int_{\{|T_pF|>K\}}|T_pF|\,d\mu_n
 \leq 2\int_{\{|F|>K/2\}}|F|\,d\mu_n
 \longrightarrow 0
 \qquad(K\to\infty).
\end{equation}
Thus the averages are uniformly integrable, and the pointwise convergence in eq.~\eqref{eq:74} holds also in $L^1$ by Vitali's theorem. Integrating the limit and using stationarity proves the claim.
\end{proof}

\section{Infinite quantum $t$-designs and their defects}\label{sec:designs}
\subsection{Background}
GKP codes \cite{Gottesman_2001} are quantum stabilizer codes for an infinite-dimensional separable Hilbert space $\CH=L^2(\R^{n})$, where the stabilizer group is chosen as a subgroup of the unitary representation of the Heisenberg group $H(\R^n)$ via displacement operators. For a wavefunction $\psi\in L^{2}(\R^n)$, displacement operators act as
\begin{align}
    D:\R^{2n}&\to U(\CH),\quad  (z_1, z_2) \mapsto D(z), \nonumber\\
    D((z_1, z_2))\psi(q) &=e^{-i\pi z_1^\intercal z_2}e^{i\sqrt{2\pi}z_2^\intercal q}\psi(q-\sqrt{2\pi}z_1).
     \label{eq:disp_def}
\end{align}
Crucially, they commute up to a phase determined by the symplectic inner product $\omega(x,y)$,
\begin{equation}
 D(x)D(y)=e^{-i\pi \omega(x, y)}D(x+y) = e^{-i2\pi \omega(x, y)}D(y)D(x)
\end{equation}
and are distributionally orthonormal
\begin{equation}
    \Tr[D(x)]=\delta(x).\label{eq:Disp_ortho}
\end{equation}
They hence define a basis for a subset of operators in $B(\CH)$ and we call self-adjoint operators that can be written as the weak integral~\cite{Hall2013,Keyl_2016}
\begin{equation}
    O_f=\int_{\R^{2n}} f(x) D(x) dx,\label{eq:Bochner}
\end{equation}
with $f\in \CS(\R^{2n})$ a Schwartz function, and $O_f\geq0$ and $\Tr O_f=1$, a physical operator. The subset of operators in $B(\CH)$ with finite trace norm $\|O\|_1=\Tr[|O|]$ is called trace-class $\CS_1(\CH)$ and the operators finite in the norm induced by the Hilbert-Schmidt inner product 
\begin{equation}
    \langle A, B \rangle_{\rm HS} = \Tr[A^{\dagger}B],\quad \|O\|_2^2=\langle O, O \rangle_{\rm HS}
\end{equation}
are Hilbert-Schmidt (HS) operators $\CS_2(\CH)$. Generalized observables are self-adjoint elements of the continuous dual of the Schwartz operators and admit the distributional representation~\eqref{eq:Bochner} with $f\in\CS'(\R^{2n})$~\cite{Keyl_2016}. All operator-valued sums and integrals below are understood by pairing with Schwartz operators.

Following~\cite{Iosue_2024} we call a measurable subset of (non-physical) states $\chi\subset \CS'(\R^n)$ with measure $\mu_{\chi}$ a quantum $t$-design if there are constants $a_{k}\in \R_{>0}$, such that
\begin{equation}
    \int_{\chi} \ketbra{\psi}^{\otimes k} d\mu_{\chi}(\psi) = a_{k}\Pi_{k} = \frac{a_{k}}{|S_{k}|}\sum_{\sigma\in S_{k}} P_{\sigma}\label{eq:t_design}
\end{equation}
is proportional to the projector onto the symmetric subspace of $\CH^{\otimes k}$ for all $k=1,...,t$. Here, $P_{\sigma}$ are unitary representations of the permutations across tensor (replica) factors,
\begin{align}
    &P_{(0)}=I,\quad &&P_{(12)}\ket{\psi}\ket{\phi}=\ket{\phi}\ket{\psi},\\
    &P_{(123)}=P_{(12)}P_{(23)}, \quad  &&P_{(132)}=P_{(23)}P_{(12)}.
\end{align}

Written in the basis eq.~\eqref{eq:Bochner}, using orthonormality of eq.~\eqref{eq:Disp_ortho}, we find

\begin{align}
    P_{(0)}&=I=D(0), \nonumber\\
    P_{(12)}&= \int_{\R^{2n}} D(x)\otimes D(-x) dx,  \nonumber\\
    P_{(123)}&= \int_{\R^{2n}\times \R^{2n}} e^{i\pi\omega(x,y)} D(x)\otimes D(y)\otimes D(-x-y) dxdy, \nonumber\\ 
    P_{(132)}&= \int_{\R^{2n}\times \R^{2n}} e^{-i\pi\omega(x,y)} D(x)\otimes D(y)\otimes D(-x-y) dxdy.\label{eq:permute_disp}
\end{align}

\subsection{GKP $t$-designs}
In \cite{Conrad_2025} we considered a set of unimodular GKP states for $\chi$. Unimodular GKP states are (non-physical quantum states) described by operators 
\begin{equation}
    G_{\Lambda, s; M}=\sum_{a\in \Z^{2n}} e^{-i\pi \phi_M(a)} e^{-i2\pi\omega(s, aM)}D(aM),
\end{equation}
for $s\in P(\Lambda)$, a fundamental domain for $\Lambda$, basis $M\in\Sp_{2n}(\R)$ for $\Lambda$, with $MJM^\intercal =J$, and $a=(a_1, a_2)$, $\phi_M(a)=a_1^{T}a_2$.
Note that the only implicit dependency on the basis $M$ for $\Lambda$ is in $\phi_M$.
Inserting into eq.~\eqref{eq:t_design}, we find
\begin{align}
    \int_{P(\Lambda)} G_{\Lambda, s; M}^{\otimes t} ds = \sum_{\{a_i\in \Z^{2n}\}_{i=1}^t}  &\lrc{\prod_{i=1}^t e^{-i\pi \phi_M(a_i)}} \lrq{\int_{P(\Lambda)} e^{-i2\pi\omega(s,\sum_{{i}} a_iM)} ds} \nonumber\\
    &\bigotimes_{i=1}^tD(a_iM).\label{eq:fundamental_int}
\end{align}
The expression in the square bracket yields generally 
\begin{equation}
    \int_{P(\Lambda)} e^{-i2\pi\omega(s,\sum_{{i}} a_iM)} ds = \delta_{0,\sum_{{i}} a_iM},
\end{equation}
such that for $t=1$, we just have
\begin{equation}
    \int_{P(\Lambda)} G_{\Lambda, s; M}^{\otimes t} ds = D(0)=I.
\end{equation}
For $t=2$, $a_1+a_2=0$ forces the expression in the curly brackets to $1$, such that
\begin{equation}
    \int_{P(\Lambda)} G_{\Lambda, s; M}^{\otimes t} ds = \sum_{\lambda \in \Lambda} D(\lambda)\otimes D(-\lambda)
\end{equation}
becomes a formal Siegel transform and only depends on $\Lambda$. Understood relative to pairing with a Schwartz operator we hence can apply the first moment formula. This yields
\begin{equation}
    \int_{Y_n}\int_{P(\Lambda)} G_{\Lambda, s; M}^{\otimes 2} ds\,d\mu_n(\Lambda) = I+\int_{\R^{2n}} D(x)\otimes D(-x)\,dx = 2\Pi_2,
\end{equation}
We have successfully reproduced the result of ref.~\cite{Conrad_2025}: the ensemble of lattices and fundamental translations forms a rigged quantum $2$-design.

Now that we are also equipped with a symplectic second moment formula, we can investigate whether the same ensemble also yields a quantum $3$-design. With eq.~\eqref{eq:fundamental_int}, we again have
\begin{equation}
    \int_{P(\Lambda)} G_{\Lambda, s; M}^{\otimes 3} ds =\sum_{\{a_i\in \Z^{2n}\}_{i=1}^{3}}  \lrc{\prod_{i=1}^{3} e^{-i\pi \phi_M(a_i)}} \delta_{a_3, -a_1-a_2}\bigotimes_{i=1}^3D(a_iM)
\end{equation}
Now, using $\phi_M(a)=aX_M a^\intercal $ for a matrix $X_M$ such that $X_M+X_M^\intercal  \equiv J \text{ (mod $2$)}$, we can evaluate the curly expression under $a_3=-a_1-a_2$ to
\begin{equation}
    \prod_{i=1}^3 e^{-i\pi \phi_M(a_i)} =e^{-i\pi a_1(X_M+X_M^\intercal )a_2^\intercal } = e^{-i\pi\omega (a_1M, a_2M)} = e^{-i\pi\omega(\lambda_1, \lambda_2)},
\end{equation}
where have set $\lambda_i=a_iM$. The right-hand side now is basis independent.

We proceed to integrate 
\begin{align}
    G_{\Lambda}^{(3)}&=\int_{P(\Lambda)} G_{\Lambda, s; M}^{\otimes 3}  ds \nonumber \\
    &= \sum_{\lambda_1, \lambda_2 \in \Lambda} e^{-i\pi\omega(\lambda_1, \lambda_2)} D(\lambda_1)\otimes D(\lambda_2)\otimes D(-\lambda_1-\lambda_2)\nonumber \\
    &= I + \sum_{\lambda_1 \in \Lambda^{\times}}  D(\lambda_1)\otimes I\otimes D(-\lambda_1) + \sum_{ \lambda_2 \in \Lambda^{\times}}  I\otimes D(\lambda_2)\otimes D(-\lambda_2) \nonumber\\
    &\hspace{1cm}+ \Theta_H^{\times}(\Lambda)
\end{align}
over $Y_n$, where we separated the zero contributions and defined the formal Siegel transform of
\begin{equation}
    H(x,y) = e^{-i\pi\omega(x, y)} D(x)\otimes D(y)\otimes D(-x-y).
\end{equation}
Applying the second moment formula now yields -- again relative to pairing with a Schwartz operator -- 
\begin{align}
    \int_{Y_n}  G_{\Lambda}^{(3)} d\mu(\Lambda) = I + P_{(13)}+ P_{(23)}+\int_{Y_n} \Theta_H^{\times}(\Lambda) d\mu(\Lambda),
\end{align}
with 
\begin{align}
    \int_{Y_n} \Theta_H^{\times}(\Lambda) d\mu(\Lambda)=P_{(132)}+ \CL(H)+\CL(\CF_2 H).
\end{align}
We need two more terms, $P_{(123)}$ and $P_{(12)}$ to complete all the ingredients for the projector $\Pi_3$. To this end consider the anti-diagonal terms in $\CL(H)$ and the diagonal terms in $\CL(\CF_2H)$. By isolating the terms, substituting, and using the definitions for $\zeta(s)$ and $\CF_2H$, one can compute
\begin{align}
    \CL(H) &= P_{(12)}+\frac{1}{2\zeta(2n)}\sum_{\substack{a,b\in \Z^{\times}\\ a+b\neq 0}} \int_{\R^{2n}} H(az,bz) dz,
\end{align}
and 
\begin{align}
    \CL(\CF_2H) &= P_{(123)}+\frac{1}{2\zeta(2n)}\sum_{\substack{a, b\in \Z^{\times}\\ a\neq b}} \int_{\R^{2n}} (\CF_2H)(az,bz) dz.
\end{align}
So in total, we obtain
\begin{equation}
    \int_{Y_n}  G_{\Lambda}^{(3)} d\mu(\Lambda) = 6\Pi_3 + \CD_3,
\end{equation}
where we have the total defect,
\begin{equation}
    \CD_3=\frac{1}{2\zeta(2n)}\sum_{\substack{a,b\in \Z^{\times}\\ a\neq -b}} \int_{\R^{2n}} H(az,bz) dz + \frac{1}{2\zeta(2n)}\sum_{\substack{a, b\in \Z^{\times}\\ a\neq b}} \int_{\R^{2n}} (\CF_2H)(az,bz) dz.
\end{equation}

\subsection{Defect analysis}
A partial sum of $\CD_3$ is the $a=b$ contribution to $\CL(H)$, which is
\begin{align}
    d_{aa}(H)&=\frac{1}{2\zeta(2n)}\sum_{a\in \Z^{\times}} \int_{\R^{2n}} D(az)\otimes D(az)\otimes D(-2az)\,dz \\ 
    &= \int_{\R^{2n}} D(z)\otimes D(z)\otimes D(-2z)\,dz. 
\end{align}
We recognize that the vector $(1,1,-2)$ classifying this defect contribution matches the defect generating line $(1,1,1)\in \F_3^3$ that prevents $t=3$ quantum designs for qutrit stabilizer ensembles in~\cite{Gross_2021}.

We can in fact exactly evaluate $\CD_3$. To this end let $\psi\in L^2(\R^{3n})$ be a wave function for the $3$-replica spaces, for which we simply write
\begin{equation}
    \psi(q)=\psi(q_1, q_2, q_3).
\end{equation}
Further we define
\begin{equation}
    r_{ab}=(a,\,b,\,-a-b)^\intercal \in {\sf A}_2,\quad Q_{ab}=r_{ab}^\intercal r_{ab}/2.
\end{equation}
It will be convenient to use compressed notation $x=r\otimes z\in \R^{6n}$, where the first tensor factor labels replica spaces with $r\in \R^3$ and the second one physical phase space $z=(z_1,\,z_2)\in \R^{2n}$.
Using the defining equation eq.~\eqref{eq:disp_def}, a short calculation then yields
\begin{equation}
    H(az, bz)\psi(q) = e^{-i2\pi Q_{ab}z_1^\intercal z_2} e^{i\sqrt{2\pi}z_2^\intercal (r_{ab}^\intercal \otimes I_n)q} \psi(q-\sqrt{2\pi}(r_{ab}\otimes z_1)).
\end{equation}
Integrating first over $z_2$, using the integral form of $\delta(Q_{ab}x)=Q_{ab}^{-n}\delta(x)$, and then evaluating the $z_1$ integral lets us write
\begin{equation}
    \int_{\R^{2n}} H(az, bz)\psi(q) dz = Q_{ab}^{-n} \psi\lr{\lrq{\lr{I_3-\frac{r_{ab}r_{ab}^\intercal }{Q_{ab}}}\otimes I_n} q}.
\end{equation}
The replica reflection and its phase-space action are
\begin{equation}
    R_{ab}=I_3-\frac{r_{ab}r_{ab}^\intercal }{Q_{ab}},\qquad
    S_{ab}=R_{ab}\oplus R_{ab}\in\Sp_6(\R)\cap\SO_6(\R).
\end{equation}
Here phase coordinates list positions before momenta, and $S_{ab}$ acts identically on each physical mode. Its Gaussian unitary has the vacuum-preserving phase convention
\begin{equation}
 U_{S_{ab}}\psi(q)=\psi((R_{ab}\otimes I_n)q)
\end{equation}
(see e.g. the survey by~\cite{Weedbrook_2012} for more on this topic).

For the second term, we have relative to physical states
\begin{equation}
    \CF_2H(az, bz)= \int_{\R^{2n}}  e^{i\pi(2b-a)\omega(z, y)} D(az)\otimes D(y)\otimes D(-az-y)\,dy,
\end{equation}
and, for $a\neq0$, expanding $P_{(12)}$ in the variable $u$ and setting $y=(a-b)z-u$, $z'=z-y/a$ gives the integrated identity
\begin{equation}
    \left[\int_{\R^{2n}} H(bz,(a-b)z)\,dz\right]P_{(12)}
    =\int_{\R^{2n}}(\CF_2H)(az,bz)\,dz.
\end{equation}
In total, reindexing lets us summarize
\begin{equation}
    \CD_3=\frac{1}{2\zeta(2n)}
    \sum_{\substack{a, b\in \Z^{\times}\\ a+b\neq0}}Q_{ab}^{-n}U_{S_{ab}}(I+P_{(12)}).
\end{equation}
The sum converges in operator norm for $n\geq 2$. Let $\rho_0$ be the vacuum state with displacement amplitudes $\Tr[\rho_0 D(x)]=e^{-\pi \|x\|^2/2}$. The vacuum is rotation invariant, so \begin{equation}
    \Tr[\rho_0^{\otimes 3} U_{S_{ab}}]=\Tr[\rho_0^{\otimes 3} P_{(12)}]=1, 
\end{equation}
and thus
\begin{equation}
    \Tr[\rho_0^{\otimes 3} \CD_3]=\|\CD_3\|=\frac{1}{\zeta(2n)}\sum_{\substack{a, b\in \Z^{\times}\\ a+b\neq0}}Q_{ab}^{-n}\eqqcolon \delta_n.
\end{equation}

To exclude the possibility of $\CD_3\propto \Pi_3$, consider the Fock state $\ket{\phi}=\ket{1,0,\hdots,0}^{\otimes 3}$. One can compute that for all $a, b \in \Z, a^2b^2(a+b)^2>0,$
\begin{equation}
    \braket{\phi | U_{S_{ab}}|\phi}={\rm Perm}(R_{ab})\in (0,1),
\end{equation}
and hence $\braket{\phi | \CD_3|\phi}<\delta_n$. 

The partial series for $(a,b) \in \lrc{(c,0), (0,c), (c,-c),\; c\in \Z^{\times}}$ each evaluate to $2\zeta(2n)$, so using the Epstein $\zeta$-function for the normalized lattice ${\sf A}_2$
we have 
\begin{equation}
    \delta_n = \frac{1}{\zeta(2n)}\sum_{(0,0)\neq (a,b)\in\Z^{2}}Q_{ab}^{-n} -6=\frac{\zeta_{{\sf A}_2}(n)}{\zeta(2n)}-6.
\end{equation}
To obtain the asymptotic estimate we count that there are exactly $6$ ways to attain the value $Q_{11}=Q_{(-1)(-1)}=Q_{1,(-2)}=\hdots =3$ and all other values for $Q_{ab}\geq 7$ for all other summands. Hence we also find $\delta_n=6\cdot 3^{-n}+O(7^{-n})$.
This concludes the derivation of Theorem~\ref{thm:main}.\qed

\section*{Acknowledgements}
We thank the organizers of the Random Geometry and Number Theory workshop at the Bernoulli center, EPFL for hosting the meeting that facilitated this collaboration.
JC is grateful to the organizers of the workshop on random quantum circuits (Cologne, 2026) for the invitation to speak at the workshop, which motivated him to produce these notes. 
We thank Nihar Gargava and Seungki Kim for valuable discussions and feedback.
JC used LLMs (ChatGPT) for studying the literature and suggesting proof strategies, proofreading and recommending improvements to the presentation in preparation of this work. Also figure~\ref{fig:ip-bound} was produced using ChatGPT.

\clearpage
\appendix
\section{Physics--math dictionary}\label{app:dictionary}
Here $\CH=L^2(\R^n)$. Below we outline a dictionary between the physics language used in this work and the corresponding mathematical object. For further reading we recommend the survey~\cite{Weedbrook_2012}, as well as the book by Hall~\cite{Hall2013}.
\begin{table}[htbp]
\centering
\small
\renewcommand{\arraystretch}{1.2}
\begin{tabular}{@{}>{\raggedright\arraybackslash}p{0.30\textwidth}>{\raggedright\arraybackslash}p{\dimexpr0.70\textwidth-2\tabcolsep\relax}@{}}
\toprule
\textbf{Physics term} & \textbf{Mathematical object} \\
\midrule
$n$ modes & The $n$ factors in $\CH\cong L^2(\R)^{\otimes n}$. \\
\addlinespace
Replica & One copy of the entire space $\CH$. $t$ replicas give $\CH^{\otimes t}$. \\
\addlinespace
Physical state & A positive trace-class operator $\rho$ with $\Tr\rho=1$. A state is pure precisely when $\operatorname{rank}\rho=1$. \\
\addlinespace
Vacuum & The normalized Gaussian $\psi_0(q)=\pi^{-n/4}e^{-\|q\|^2/2}$. \\
\addlinespace
Displacement operator & A Weyl operator $D(x)$, representing a phase-space translation $x\in\R^{2n}$. \\
\addlinespace
GKP state & A generalized common eigenvector of a commuting family of lattice-indexed Weyl operators. \\
\addlinespace
Qudit & A quantum system with Hilbert space $\C^p$. \\
\addlinespace
Gaussian unitary & A unitary implementing a linear symplectic transformation of phase space. \\
\addlinespace
Replica permutation $P_\sigma$ & The operator permuting the $t$ factors of $\CH^{\otimes t}$. \\
\addlinespace
Schwartz operator & An operator with a Schwartz integral kernel, used to test the weak moment identities (see ref.~\cite{Keyl_2016}). \\
\bottomrule
\end{tabular}
\caption{Overview of physics-math dictionary relevant to this work.}
\label{tab:physics-math-dictionary}
\end{table}

Generalized state vectors are members of $\CS'(\R^n)$ in the rigged Hilbert space,
\begin{equation}
   \CS(\R^n)\subset L^2(\R^n)\subset\CS'(\R^n),
\end{equation}
where the Schwartz operators form the separate test space for operator-valued identities.

\bibliographystyle{apalike}
\bibliography{references}

\end{document}